\documentclass[11pt]{article}

\usepackage{amsmath,amssymb}
\usepackage{enumerate}
\usepackage[hyphens]{url}
\usepackage{hyperref}
\usepackage{xcolor}
\hypersetup{colorlinks,linkcolor={blue!50!black},citecolor={red!50!black},urlcolor={blue!70!black}}
\usepackage{mathtools}
\usepackage[hyphenbreaks]{breakurl}
\usepackage{graphicx}
\usepackage{graphbox}
\usepackage{microtype}
\usepackage{siunitx}
\usepackage{booktabs}
\usepackage{cleveref}
\usepackage{autonum}
\usepackage{breakurl}
\usepackage{todonotes}
\usepackage{verbatim}
\usepackage{enumerate}
\usepackage{mathtools}
\usepackage{float}
\usepackage{caption}
\usepackage{subcaption}
\usepackage{authblk}

\usepackage[margin=1.3in]{geometry}

\usepackage{amsthm,thmtools}

\declaretheorem[name=Theorem,numberwithin=section]{theorem}
\declaretheorem[name=Corollary,numberlike=theorem]{corollary}
\declaretheorem[name=Lemma,numberlike=theorem]{lemma}
\declaretheorem[name=Proposition,numberlike=theorem]{proposition}

\declaretheorem[name=Definition,style=definition,numberlike=theorem,qed={$\lrcorner$}]{definition}

\declaretheorem[name=Remark,style=definition,numberlike=theorem,qed={$\lrcorner$}]{remark}

\renewcommand{\leq}{\leqslant}
\renewcommand{\geq}{\geqslant}

\newcommand{\indep}{\perp \!\!\! \perp}
\newcommand{\EE}{\mathbb{E}}
\newcommand{\R}{\mathbb{R}}

\newcommand{\cA}{\mathcal{A}}

\newcommand{\cP}{\mathcal{P}}

\newcommand{\cX}{\mathcal{X}}

\newcommand{\bx}{\mathbf{x}}
\newcommand{\bz}{\mathbf{z}}

\newcommand{\bt}{\mathbf{t}}
\newcommand{\bs}{\mathbf{s}}

\newcommand{\bX}{\mathbf{X}}
\newcommand{\bY}{\mathbf{Y}}

\newcommand{\by}{\mathbf{y}}

\newcommand{\TAf}[1]{\mathrm{T}({#1})}

\newcommand{\definedas}{\coloneqq}

\newcommand{\E}{\mathbb{E}}

\newcommand{\Le}{\mathcal{L}}
\newcommand{\Lee}{L}

\newcommand{\lB}{a^\text{left}}
\newcommand{\rB}{a^\text{right}}

\newcommand{\lt}{\text{left}}
\newcommand{\rt}{\text{right}}

\newcommand{\Law}[1]{\operatorname{Law}(#1)}

\newcommand{\Paths}{\operatorname{Paths}}
\newcommand{\Tracks}{\operatorname{Tracks}}

\newcommand{\paths}[2]{\operatorname{Paths}({{#1},{#2}})}

\newcommand{\tracks}[2]{\operatorname{Tracks}({{#1},{#2}})}

\newcommand{\Hopf}{\mathcal{H}}
\newcommand{\HopfInv}{\Hopf^\star}

\title{Proper Scoring Rules, Gradients, Divergences, and Entropies for Paths and Time Series}
\author{Patric Bonnier\thanks{\textit{bonnier@maths.ox.ac.uk}} }
\author{Harald Oberhauser\thanks{\textit{oberhauser@maths.ox.ac.uk}}}
\date{}
\affil{Mathematical Institute, University of Oxford}

\begin{document}
	\maketitle
	
	\begin{abstract}
    Many forecasts consist not of point predictions but concern the evolution of quantities.
    For example, a central bank might predict the interest rates during the next quarter, an epidemiologist might predict trajectories of infection rates, a clinician might predict the behaviour of medical markers over the next day, etc.
    The situation is further complicated since these forecasts sometimes only concern the approximate ``shape of the future evolution'' or ``order of events''.
    Formally, such forecasts can be seen as probability measures on spaces of equivalence classes of paths modulo time-parametrization.
    We leverage the statistical framework of proper scoring rules with classical mathematical results to derive a principled approach to decision making with such forecasts.
    In particular, we introduce notions of gradients, entropy, and divergence that are tailor-made to respect the underlying non-Euclidean structure.
  \end{abstract}
	
	\section{Introduction}
	Scoring rules provide a principled approach to form and evaluate probabilistic predictions.
  The earliest applications go at least back to the evaluation of weather forecasts \cite{brier1950verification}, but have since then developed into a rich theoretical framework that plays a central part in modern statistical inference.
  We refer to \cite{gneiting2007strictly} and \cite{dawid2007geometry} for a general background.  
  The theoretical underpinnings of scoring rules are well-developed, but nearly all of the literature focuses on prediction of vector-valued or scalar-valued quantities.  
  The aim of this article is develop a scoring rule framework for an important class of non-Euclidean data, namely sequential data -- both in discrete and continuous time.

\paragraph{The Drawbacks of (Na\"ive) Vectorization.}  Given a dataset consisting of multi-variate time series (TS), a common approach is to flatten each TS into a long vector and then use a standard pipeline for vector-valued data.
  However, this approach has several drawbacks.
  Firstly, there's trouble whenever the different TS are irregularly sampled or are of different length since this embeds the different TS in Euclidean spaces of different dimension.
  Typically this addressed by ad-hoc approaches such as adding synthetic data by interpolation or dropping data points.
  Secondly, often the relevant information is independent of the time-parametrization (``\emph{time-warping invariance}''), at least to a large degree; for instance, the meaning of a spoken word or an object being filmed are both independent of how fast or slow the audio signal or video is presented.
  Finally, many models are naturally formulated in continuous time rather than discrete time; for example, stochastic differential equations form a popular class of models in many applications and it is unclear how to evaluate such continuous time models in a scoring rule framework besides above naive vectorization on an arbitrary time grid. 
  
  \paragraph{A Non-Euclidean Data Domain.}
Key to our approach is that classic tools from pure mathematics faithfully capture the Non-Euclidean structure of the space of (unparametrized) paths.
While there is no linear structure that allows for addition of paths of different length, any two paths can be concatenated into one path and any path can be run backwards.
Both these operations -- \emph{concatenation and reversal} -- are independent of the choice of parametrization, hence they also apply to equivalence classes of paths under reparametrization.
We refer to an unparametrized path -- that is an equivalence class of paths under reparametrization -- as a \emph{track}, as it is defined uniquely by the track it carves out in the space where it evolves.
A classical result \cite{chen-58} is that there is a \emph{``feature map''} from the set of tracks into a linear space that is \emph{functorial} and \emph{universal}; the former means that operations on tracks (concatenation and reversal) turn into algebraic operations in feature space, the latter means that any function of tracks can be approximated as a linear functional of this map.  
Moreover, this map is given as a series of iterated integrals which makes it amenable to computation and we refer to it as the \emph{signature map}.
In fact, the co-domain of this feature map (``\emph{the feature space''}) is not only a linear space but forms a so-called Hopf algebra and it is the Hopf algebra structure that captures operations on tracks as algebraic operations. 
The third mathematical ingredient that we use are gradients of functions of
tracks: the usual definition of linear (Fr\'echet) differentiability can sometimes be unsuitable for such functions due to the above lack of linear structure.
However, Pansu generalized classical differentiation to a special class of groups and we leverage this to define \emph{gradients of functions of (unparametrized) paths}.
We show that the usual guarantees of gradient descent algorithms apply which allows us to compute quantities associated with our scoring rule framework (as even in the case of Euclidean data, many quantities are not given in closed form, but can be found by first order methods). 

\paragraph{Outline.}
Section~\ref{sec:scoring} recalls the basic definitions and general of scoring rules. 
Section~\ref{sec: structure of tracks} contains the theoretical background; it formalizes the structure of the spaces of (unparametrized paths), and introduces the signature feature map and the Hopf algebra structure of its co-domain (the feature space). 
Section \ref{sec:scoring tracks} then shows how these quantities lead to natural scoring rules on the non-Euclidean space of tracks; in particular the so-called antipode of the Hopf algebra plays key role to relate the scoring rule framework to structural properties of tracks. 
From general results this then immediately leads to definitions of entropy, divergence, and mutual information that -- unlike the (na\"ive) vectorization approach outlined above -- are compatible with the structure of (probability measures) on spaces of (unparametrized) paths. 
Section~\ref{sec:gradients} then utilizes that one may identify a track as ``group-like'' element and shows that the concept of Pansu differentiabilty leads to a natural notion of gradient descent on the space of paths resp. tracks.
Finally, Section~\ref{sec:experiments} demonstrates that despite this approach being motivated by pure maths, the resulting quantities lead to efficiently computable quantities that have some advantageous properties compared to other methods with similar invariances.

\subsection{Related Work}
One of the earliest empirical insights for time series data was that time-parametrization (``time warping'') invariance is of key importance \cite{SakoeChiba71, 1163055}.
Arguably the most popular way to address this invariance is via the classical dynamic time warping distance (DTW) and its many variations that introduce a distance between time series by searching over time changes.
For example, \cite{Cuturi2017SoftDTWAD, Blondel2021DifferentiableDB} introduce a regularised version of DTW, so-called soft DTW (s-DTW) which addresses the fact that the DTW distance is not differentiable, making it viable fo use in deep learning pipelines. In the process of doing so it loses the invariance that DTW enjoys and introduces a trade-off of smoothness versus invariance.
A more general point is that DTW and its variations do not aim to provide the full forecasting framework of scoring rules (divergence between measures, entropy, mutual information of TS) and although DTW approaches successfully deal with time-parametrization, they ignore other structural properties such as concatenation and reversal of TS.
Another drawback is that while the focus of DTW is on discrete time it can be formulated in continuous time (so-called Fr\'echet distance) but the computation scales with quadratic complexity in the number of sequence entries which makes it too expensive for many sources of high-frequency data, whereas our distance can be computed in linear time for the price of higher complexity in the state space dimension.
Ultimately the reason for this increase in efficiency is that the time-warpings are never explicitly computed or exhibited. 

Another area that is directly related is kernel learning.
Any kernel $k:\cX \times \cX \to [0,\infty)$ with reproducing kernel Hilbert space $H_k$ induces a scoring rule by setting $s(x,\mu)\coloneqq \| \delta_x - \mu\|_{H_k}^2$ where $\|\mu\|_{H_k}^2 = \int \int k(x,y) \mu(dx)\mu(dx)$, see \cite[Section 5]{Gneiting2007StrictlyPS} and several kernels for sequences have been developed in the literature.
Most relevant to our approach is the ``signature kernel'' introduced in \cite{2016arXiv160108169K}. 
However, for any scoring rule given by a kernel, the (generalized) divergence becomes simply the maximum mean discrepancy and the entropy simply the variance in the RKHS $H_k$. 
While kernels give rise to a powerful class of scoring rules, the success and popularity of non-kernel based scoring rules on $\cX=\R^n$ is motivation enough to look for other interesting, non-linear scoring rules.

The technical key to our approach comes from mathematics where iterated integrals, so-called signatures, and non-commutative algebras are use to represent paths.
This goes at least back to seminal work of Chen~\cite{chen-58} in algebraic topology and subsequent applications in control theory \cite{fliess1981fonctionnelles, RogerW1976167} and more recently rough path theory \cite{MR2314753}.
These results have been influential in stochastic analysis \cite{lyons-qian-02,friz2014course} and only more recently have been started to be explored in a statistical and machine learning context.     
We mention pars-pro-toto \cite{papavasiliou2011parameter,ChevyrevOberhauser18, distrReg} for inference about laws of stochastic processes; \cite{2016arXiv160108169K,salvi2021higher,fermanian2021framing} for kernel learning; \cite{toth2020bayesian,dyer2021approximate,lemercier2021siggpde} for Bayesian approaches; \cite{ni2020conditional,buhler2020datadriven,kidger2021neural} for generative modelling; \cite{giusti2021signatures,lee2020path,Chevyrev2018PersistencePA} for applications in topology; \cite{diehl2020generalized,Diehl_2020,toth2021seqtens} for algebraic perspectives.

Finally, we mention that the two topics that are central to us -- invariances and non-Euclidean structure -- have been considered in different contexts in scoring rule frameworks.
For example, \cite{Fissler2017OrderSensitivityAE} studies equi- and in-variances for scoring rules for Euclidean data; non-vector valued data such as sets, contours, intervals, and quantiles have received attention \cite{bolin2015excursion,molchanov2005theory, fissler2021forecast}.

  \section{Proper Scoring Rules, Entropies, and Divergences} \label{sec:scoring}
  We briefly recall general background on scoring rules following closely the notation in \cite{dawid2007geometry}, see also  \cite{gneiting2007strictly,Dawid2014TheoryAA,Grunwaldidid2004,GKMO2018}. 
  Let $ \cX $ be a measurable space (\emph{the outcome space}), $ \cA $ be a set (\emph{the action space}), and $ \Le : \cX \times \cA \to \R $ be a function  (\emph{the loss function}).
  Further, let $X$ be a $\cX$-valued random variable.
  We consider the following game between a Decision-maker and Nature: the task of the decision maker is to choose an action $a \in \cA$, after which Nature
  reveals the outcome $x \in \cX$ that is given by sampling $X$.
  The decision maker then suffers the loss $\Le(x,a)$.
 
  Given a set $\cP$ of probability measures on $\cX$, a principled probabilistic (Bayesian) approach to this decision problem is to proceed as follows:

	\begin{enumerate}[(I)] 
  \item \label{step: bayes act}Associate with every $\mu \in \cP$ its \emph{Bayes act} $ a_\mu \in \cA $ defined by
		\begin{align}
      a_\mu \definedas \arg\min_{a\in\cA} \E_{X\sim \mu}[ \Le(X,a)]
		\end{align}
    (assuming that a minimum exists; if it is not unique, choose $a_\mu$ arbitrary among the minimizers).
		\item\label{step: scoring rule} Use the Bayes act $a_\mu$ to define the \emph{scoring rule} $s:\cX \times \cP \to \R$ on $ \cX $ given by 
		\begin{align}\label{def:scoring rule}
		s(x,\mu) \definedas \Le(x,a_\mu).
		\end{align}
  \item \label{step: entropy}Use the scoring rule $s$ to define the (generalised) \emph{entropy} $H$, the (generalised) \emph{divergence} $d: \cP \times \cP \to \R$, and the (generalised) \emph{mutual information} $I: \cP \times \mathcal{Q} \to \R$ as
    \begin{align}
      H: \cP &\to \R,\quad      \mu \mapsto \E_{X\sim \mu}[ s(X,\mu)],\\
      d: \cP \times \cP &\to \R,\quad     (\nu,\mu) \mapsto \E_{X\sim \nu}[ S(X,\mu)]-H(\nu),\\
      I: \cP \times \mathcal{Q} &\to \R,\quad   (\mu,\nu)  \mapsto H(\mu) - \E_{U \sim \nu}[ H(\mu|U)].
    \end{align}
    where $\mathcal{Q}$ denotes a space of probability measures (not necessarily on $\cX$) and $\mu|U$ denotes the law of $\mu$ conditioned on the random variable $U$. 
	\end{enumerate}%
  The above definitions and nomenclature is justified as follows:
  firstly, it is an instructive exercise to check that for standard choices of state space $\cX$, action space $\cA$, and loss function $\Le$, the above reduce to classical definitions of entropy, divergence and mutual information (e.g.~if $\cA$ is the set of densities on $\R^n$ then the \emph{log score} $L(x,a) \coloneqq -\log a(x)$ from \cite{good52} yields the the usual Shannon entropy, Kullback-Leibler divergence, and mutual information); see~\cite{dawid2007geometry} for more examples.
  Secondly, Theorem \ref{thm:scoring} below shows that characteristic properties hold in the full generality of the above setup: 
  \begin{theorem}[\cite{dawid2007geometry}]\label{thm:scoring}
    Let $\cX$, $\cA$, and $L$ be as above.
    Further, denote with $H$, $I$, and $d$ the associated (generalized) scoring rule, entropy, and divergence.
    Then
    \begin{enumerate}
		\item\label{itm:proper} the scoring rule \eqref{def:scoring rule} is \emph{proper}, that is
      \begin{align}
       \mu \mapsto \E_{X \sim \nu}[s(X,\mu)] 
      \end{align}
      is minimized at $\mu=\nu$.
		\item $ \mu \mapsto H(\mu) $ is concave,
		\item $(\nu, \mu) \mapsto \E_{X \sim \nu}[s(X,\mu)] $ is affine in $ \nu $ for every $\mu$,
		\item $ d(\mu,\nu) \geq 0 $ with equality for $ \mu=\nu $,
     \item $I(\mu,\nu) \ge 0$ with equality for $\mu \indep \nu$.
	\end{enumerate}
\end{theorem}
Different applications areas demand different scoring rules.
Classical choices for the Euclidean case $\cX=\R^n$ are besides the already mentioned log score, the Brier score, the Tsallis score, the Bregman score, the Hyv\"arinen score, etc.; see~\cite{dawid2007geometry} for details. 
The aim of the remainder of this article is to study the case when $\cX$ is a space of paths or a space of equivalence classes of paths (under reparametrisation).  
  \paragraph{A Toy Example: From Feature Maps to Bayes Actions.}
  To motivate our scoring rule for paths let us first revisit the vector-valued case, $\cX= \R^n$.
  To arrive at a proper scoring rule, the space of Bayes actions $\cA$ should be large enough to characterize any (sufficiently nice) probability measures on $\cX=\R^n$.
  A classic way to characterize a probability measure, is to consider the sequence of moments,
  \begin{align}\label{eq:moments}
    (1, \E[X], \E[X^{\otimes 2}], \E[X^{\otimes 3}],\ldots) 
  \end{align}
  that is, $\E[X]$ is the mean vector, $\EE[X^{\otimes 2}]$ is the covariance matrix, etc. 
  (We tacitly assume that the sequence of moments is well-defined and decays quickly enough so that it characterizes the measure, see Remark \ref{rem:moments characteristic}).
  The sequence \eqref{eq:moments} is an element of the set
  \begin{align}\label{eq:define tensor algebra}
    \Hopf \coloneqq \prod_{m \ge 0} (\R^n)^{\otimes m}
  \end{align}
  of sequences of tensor of increasing degree $m$.
  In fact, this set $\Hopf$ forms a vector space by element-wise addition of tensors of the same degree.
  If we define the ``feature map''
  \begin{align}
    \varphi: \R^n \to \Hopf, \quad x \mapsto (1, x, x^{\otimes 2}, \ldots).
  \end{align}
  With the above notation, the moment sequence \eqref{eq:moments} is simply the mean of $\varphi(X)$,
  \begin{align}\label{eq:expected feature}
    \E[\varphi(X)] =    (1, \E[X], \E[X^{\otimes 2}], \E[X^{\otimes 3}],\ldots) \in \Hopf
  \end{align}
 Using a well-known characterization of the mean as minimizer of a quadratic 
  we can introduce the loss function
  \begin{align}
   \Le(x,a) \coloneqq \| \varphi(X) - a\|^2  
  \end{align}
  which associates with a probability measure $\mu$ on $\R^n$ the Bayes action
  \begin{align}
    a_\mu \equiv \E[\varphi(X)] \equiv \underset{m \in \Hopf}{\operatorname{argmin}} \E[ \Le(X,a)].
  \end{align}
  It follows from general principles that the resulting scoring rule on the state $\cX=\R^n$ and action space $\cA = \Hopf$,
  \begin{align}
   s(x,\mu) = \Le(x,a_\mu) 
  \end{align}
  is proper. 
  Despite the elementary nature of this example it gives us a simple way to associate with any ``feature map'' a Bayes action and a scoring rule and already simple variations lead to interesting questions: for example, if the quadratic loss function is replaced by the absolute value, one ends with medians of moment as Bayes action and many other choices are possible.
  Such questions fall under the framework of ``elicitation'' of properties of probability measures with scoring rules which is an active research area, already in the classical vector-valued (even scalar) case; see \cite{savage1971elicitation} and \cite{Abernethy2012ACO,Steinwart2014ElicitationAI,Frongillo2015VectorValuedPE} for some of the recent advances.
  \begin{remark}\label{rem:moments characteristic}
    The question which probability measures on $\R^n$ are characterized by the moment sequence \eqref{eq:moments} is classical but quite subtle in general.
    But for compactly supported measures this trivially holds.
    Our focus will soon shift to probability measures on pathspace but since spaces of paths are generically not even locally compact, compactness is a too strong assumption; in fact, important examples of measures on pathspace such as geometric Brownian motion are not characterized by their ``signature moments'' that we will use in Section~\ref{sec: structure of tracks} and Section~\ref{sec:scoring tracks}.  
    However, one can replace the moment sequence~\eqref{eq:moments} by a normalized moment sequence that characterizes any probability measure on $\R^n$ and this extends to path space and signature moments, see \cite{ChevyrevOberhauser18} for details. 
    Hence, for simplicity, we assume throughout that the probability measures are characterized by their expected feature map (since this is possible by a slight modification of the feature map). 
  \end{remark}
  \section{Structure of the Space of (Unparametrized) Paths}\label{sec: structure of tracks}	

  We review classic mathematical results about spaces of paths going back to seminal work of Chen \cite{chen-58}. 
  The main result is the existence of a \emph{``feature map''}
  \begin{align}\label{eq:signature}
  \Phi: \Tracks \to \Hopf,\,\quad\bx \mapsto \Phi(\bx) 
  \end{align}
  that has as domain the set $\Tracks$ that consist of equivalence classes of paths that evolve in $\R^n$, and as co-domain the linear space $\Hopf$.
 
  We already encountered $\Hopf$ in Section~\ref{sec:scoring} where it arose as the vector space of sequences of tensors $\bt_m \in (\R^n)^{\otimes m}$ of increasing degree $m$, see \eqref{eq:define tensor algebra}.
  However, $\Hopf$ is not only a vector space but a so-called Hopf algebra: we can multiply elements of $\Hopf$ and take the ``inverse'' of elements of $\Hopf$. 
  One of the well-known and attractive properties of the map~\eqref{eq:signature} is that these two algebraic operations (multiplication and inversion) capture the natural operations on $\Tracks$ (conatentation and reversal). 
  Exploiting this correspondence will be essential for our main results in Section~\ref{sec:scoring tracks}.  
  \paragraph{The Domain of Paths.}
  A bounded variation path\footnote{For simplicity we focus on (equivalence classes of) bounded variation paths but all the results immediately extend to paths with much less regularity such as trajectories of stochastic differential equations or (fractional) Brownian motion by replacing the iterated Riemann--Stieltjes integrals by stochastic integrals or rough path integrals \cite{MR2314753}} $\bx$ in $\R^n$ is a continuous map 
  \begin{align}
    \bx : [0,T] \to \R^n\, \text{ such that } \|\bx\| \coloneqq \sup_{(t_1,\ldots,t_L): 0 \le t_1 < \ldots < t_L < T} \sum  \| \bx(t_{i+1})- \bx(t_i) \| < \infty. 
  \end{align}
  For $a,b \in \R^n$ we denote with $\paths{a}{b}$ the set of all continuous bounded variation paths that start at $a$ and end in $b$,
  \begin{align}
    \paths{a}{b} \coloneqq \{\bx | \bx:[0,T] \to \R^n\text{ is of bounded variation, } T>0, \bx(0)=a, \bx(T)=b\}. 
  \end{align}
  and by
  \begin{align}
   \Paths \coloneqq \bigcup_{a,b \in \R^n} \paths{a}{b} 
  \end{align}
  the set of all bounded variation paths in $\R^n$.
  Although $\Paths$  is not a linear space, it has a rich structure given by concatenation and time reversal.
  Informally, this says that if one can go from $a$ to $b$ and from $b$ to $c$ then one can go from $a$ to $c$ and that if one can go from $a$ to $b$ then one can go from $b$ to $a$.
  Formally, concatenation and reversal are defined as 
  \begin{enumerate}
  \item For $\bx \in \paths{a}{b}$, $\by \in \paths{b}{c}$, their \emph{concatenation} $\bx \star \by \in \paths{a}{c}$ is defined by 
    \begin{align}
      (\bx \star \by)(t) \coloneqq
      \begin{cases}
        \bx(t) \text{, if } t \in [a,b]\\ 
        \bx(b)-\by(b)+\by(t) \text{, if } t \in [b,b + c]\\ 
      \end{cases}
    \end{align}
  \item for any $ \bx\in \paths{a}{b}$ there exists an \emph{inverse} path $\overleftarrow \bx \in \paths{b}{a}$ defined as
    \begin{align}
      \overleftarrow\bx(t) \coloneqq \bx( b - t).
    \end{align}
  \end{enumerate}
 \paragraph{The Domain of Tracks} 
  As discussed above, often we want to ignore the time parametrization, hence the fundamental object we care about is not the set of paths but equivalence classes of paths.
  It turns out that is useful to work with slightly more general equivalence relation, namely that of \emph{tree-like equivalence} $\sim$.
    We define
  \begin{align}
    \tracks{a}{b} \coloneqq \paths{a}{b} / {\sim} \text{, and } \Tracks\coloneqq \Paths / {\sim} 
  \end{align}
  With slight abuse of notation, we use the same notation $\bx$ for an element of $\Paths$ and an element of $\Tracks$ but emphasize that an element $\bx \in \Tracks$ is a whole equivalence class of paths. 
  We give the precise definition of the equivalence relation $\sim$ in Appendix \ref{app:TLE} and only note here that if two paths $\bx:[0,T] \to \R^n,\by:[0,S] \to \R^n$ differ by time-parametrization, that is $\bx(t)=\by(\varphi(t))$ for every $t$ and an increasing function $\varphi:[0,T] \to [0,S]$, then $\bx \sim \by$.
  However, in addition to time parametrization, tree-like equivalence also identifies paths that backtrack all their excursions, see Appendix \ref{app:TLE}.
  We invite readers to think of elements of $\Tracks$ like animal tracks in nature: they provide shape and direction but not the speed at which the track was made.
  In particular, we note that the above operations of concatenation and reversal are well-defined for the elements of $\Tracks$; after all, they do not depend on the time-parametrization. 
  So again, with slight abuse of notation we have concatenation and reversal map,
  \begin{align}
    \star : \tracks{a}{b} \times \tracks{b}{c} \to \tracks{a}{c} \text{ and } \overleftarrow{\bullet}: \tracks{a}{b} \to \tracks{b}{a}.
  \end{align}

\paragraph{The co-domain $\Hopf$.} 
Our first encounter of $\Hopf$ was in Section~\ref{sec:scoring} as the state space of the moment map~\eqref{eq:define tensor algebra}.
However, a more abstract way to introduce is by identifying it as the free algebra over $\R^n$.
Informally, this means we want to keep the vector space structure of $\R^n$ but we also would like to have a multiplication and do this in the most general way possible.
Formally, this means $\Hopf$ is the free algebra over $\R^n$.
Despite this abstract characterization as a free object, the space $\Hopf_n$ has a very concrete form which we will take as its definition,
\begin{align}
  \Hopf \coloneqq \prod_{m \geq 0} (\R^n)^{\otimes m}  \coloneqq \{ \bt=(\bt_0,\bt_1, \bt_2,\bt_3, \ldots) : \bt_m \in (\R^n)^{\otimes m} \}.
\end{align}
(one can then directly verify that this indeed is the free algebra, see~\cite{reutenauer-93}. 
That is, an element $\bt$ of $\Hopf$ is sequence of tensors $(\bt_0,\bt_1,\bt_2,\ldots)$ of increasing degree $m$ where by convention $(\R^n)^{\otimes 0} = \R $.
The vector space structure of $\Hopf$ is simply given as element-wise addition: addition of $\bs,\bt \in \Hopf$ is defined as 
\begin{align}
  \bs + \bt = (\bs_0 + \bt_0, \bs_1 + \bt_1, \ldots)  
\end{align}
and their multiplication is defined by $ (\bs \cdot \bt)_m = \sum_k \bs_k \otimes \bt_{m-k} $, i.e
\begin{align} \label{eq:multiplication}
  \bs \cdot \bt \coloneqq (1, \bs_1+ \bt_1, \bs_2 + \bs_1 \otimes \bt_1 + \bt_2, \ldots)
\end{align}
where $\otimes$ denotes the usual tensor (outer) product.
Like matrix multiplication, this multiplication is associative but in general not commutative, $\bs \cdot \bt \neq \bt \cdot \bs$ and it has as multiplicative unit $(1,0,\ldots) \in \Hopf$,
\begin{align}
 \bt \cdot (1,0,0,\ldots) = (1,0,\ldots) \cdot \bt = \bt.
\end{align}
The existence of a unit for multiplication naturally leads to the question of the existence of inverses, that is for $\bt \in \Hopf$ can one find another element in $\Hopf$, denoted by $\bt^{-1} \in \Hopf$, such that 
\begin{align}
  \bt \cdot \bt^{-1} = \bt^{-1} \cdot \bt = (1,0,0,\ldots).
\end{align}
This is true whenever $ \bt_0 \not= 0 $, and moreover, $ \bt \mapsto \bt^{-1} $ has the explicit formula
\begin{align}
\bt^{-1} = \frac1{\bt_0} \Big\{\sum_{m\geq 0} (1-\frac1{\bt_0}\bt)^{\otimes m} \Big\}.
\end{align}

\paragraph{The Feature map $\Phi:\Tracks \to \Hopf$.}
\begin{definition}\label{def:iterated integrals}
  For $\bx \in \Paths$, $\bx:[0,T] \to \R^n$ define 
  \begin{align}
   \int d\bx^{\otimes m} \coloneqq  \int_{0 \le t_1 <\cdots < t_m \le T} d\bx(t_1) \otimes \cdots d \bx(t_m) = \int \dot \bx(t_1) \otimes \cdots \dot \bx(t_m) dt_1 \cdots dt_m. 
  \end{align}
\end{definition}
It is known that if two paths $\bx,\by \in \Paths$ are tree-like equivalent, $\bx \sim \by$, then $\int d\bx^{\otimes m}= \int d\by^{\otimes m}$ for every $m \ge0$, see~\cite{MR2630037}.
In fact, for the case of reparametrization $\bx(t)=\by(\tau(t))$ this follows immediately from the change of variables formula.
With slight abuse of notation we now define $\int d \bx^{\otimes m}$ for $\bx\in \Tracks$.
\begin{definition} For $\bx \in \Tracks$, define 
  \begin{align}
    \int d\bx^{\otimes m} \coloneqq  \int d\bx^{\otimes m}_{\text{path}} %
  \end{align}
where $\bx_{\text{path}} \in \Paths$ is in the equivalence class of $\bx$ and $\int d\bx^{\otimes m}_{\text{path}}$ is as in Definition~\ref{def:iterated integrals}.  
\end{definition}
By~\cite{MR2630037} $\int d\bx^{\otimes m}$ for $\bx \in \Tracks$ is well-defined in the sense that the choice of $\bx_{\text{path}}$ does not matter. 
We refer to the resulting map as the signature map (this is also known as the Chen--Fliess series or chronological exponential).
\begin{definition}
  We call
  \begin{align}
   \Phi: \Tracks \to \Hopf,\quad \bx \mapsto \left( \int d\bx^{\otimes m} \right)_{m \ge 0} 
  \end{align}
  the signature map.
\end{definition}
A well-known key property of the map $\Phi$ is that %
concatenation and reversal in $\Tracks$ correspond to multiplication and inversion in $\Hopf$.
Further, the map $\Phi$ is universal (up to fixing the starting point of the track, which is why we fix the starting point $a \in \R^n$ and restrict to the domain $\bigcup_{b\in\R^n} \Tracks(a,b)$) in the sense that it linearizes continuous functions on $\Tracks$.
We summarize all this in Theorem~\ref{thm:signature} below.

 \begin{theorem}\label{thm:signature}
   For every $a \in \R^n$ the map
   \begin{align}
     \Phi: \bigcup_{b \in \R^n}\tracks{a}{b}\to \Hopf , \quad \bx \mapsto \left(\int d\bx^{\otimes m}\right)_{m \ge 0} 
   \end{align}
   is injective and 
   \begin{enumerate}
   \item \label{itm:concat2}
           $\Phi(\bx\star \by) = \Phi(\bx) \cdot \Phi(\by)$
 \item\label{itm:inverse} 
           $\Phi(\overleftarrow \bx) = \Phi(\bx)^{-1}$
         \item \label{itm:universal}for every $f \in C(\Tracks,\R)$, $\epsilon>0$ there exists a linear functional $\ell \in \Hopf^\star$ such that 
           \begin{align}
             | f(\bx) - \langle \ell, \Phi(\bx) \rangle| < \epsilon
           \end{align}
uniformly in $\bx$ on compacts.
         \end{enumerate}
 \end{theorem}
 \begin{proof}
   This is a folk theorem in algebraic topology and control theory; see \cite{chen-58} and \cite{fliess1981fonctionnelles}. 
   What is less standard is that we use the treelike equivalence from $\sim$ from~\cite{MR2630037}. 
 \end{proof}
 \begin{remark}
   \label{rm:compute signature}
   The space $\Hopf \equiv \prod_{m \ge 0} (\R^d)^{\otimes m}$ is graded by the tensor degree $m$, and %
   $\Phi(\bx) \equiv (\int dx^{\otimes m})_{m \ge 0}$ decays exponentially fast in $ m $, that is
   \begin{align}
    \|{\int dx^{\otimes m}} \|\le \frac{\|\bx \|^m }{m!}
   \end{align}
   (on the right hand side $\|\bullet\|$ denotes the bounded variation (semi-)norm, on the left-hand side it denotes the norm on $(\R^d)^{\otimes m}$).  
  Hence, in practice one only needs to compute the first $m$ iterated integrals of $\Phi(\bx)$.  
  For piecewise linear tracks $\bx \in \Tracks$ -- which is how we identify time series -- the first $m$ entries of the map $\Phi(\bx)$ can be computed in $O(L d^M)$ computational steps: if a track is given by piecewise linear segments $v_1,\ldots,v_L \in \R^n$ then
   \begin{align} \label{eq:expsig}
     \left(\int d\bx^{\otimes m}\right)_{m \in 0,1,\ldots,M} = \exp(v_1) \cdots \exp(v_L), 
   \end{align}
   where $\exp(v) \coloneqq (1, v, \frac{v^{\otimes 2}}{2!}, \ldots) \in \Hopf$.
   Hence, for a low dimensional state space, the map $\Phi(\bx)$ can be approximately in time that scales linearly in the length of the path.
 \end{remark}

\paragraph{The Antipode in $\Hopf$.} 
The two operations of addition $\bs + \bt$ and multiplication $\bs \cdot \bt $ turn $\Hopf$ into a (non-commutative) algebra $(\Hopf, +, \cdot)$.
However, $ \Hopf $ comes with a bit more structure, namely 
the so-called \emph{antipode} map 
\begin{align}
  \alpha : \Hopf \to \Hopf 
\end{align}
which is defined as the linear function given by linear extensions of the map
\begin{align}
  (\R^n)^{\otimes m} \to (\R^n)^{\otimes m}, \quad v_1 \otimes \cdots v_m \otimes \mapsto  (-1)^m v_m \otimes \cdots \otimes v_1.  
\end{align}
There is an important subset $ G $ of $ \Hopf $ defined by the property
\begin{align}
G := \{ g \in \Hopf : \alpha(g) = g^{-1} \}.
\end{align} 
It turns out that $ G $ in fact forms a group and will play an important role in
our Bayes acts for the simple fact that the feature map $ \Phi $ takes values in
$ G $. We summaries this along with some facts about $ \alpha $ that we use later in the following lemma.
\begin{lemma} \label{lem:antipode}
  Let $\alpha$ be the antipode on $\Hopf$.
  Then
	\begin{enumerate}
		\item \label{itm:a1}$ \alpha^2 = 1 $,
		\item \label{itm:a2}$ \alpha(\bs \cdot \bt) = \alpha (\bt )\alpha (\bs) $,
 \item \label{itm:group}
      If $ \bx \in \Tracks $, then $ \Phi(\bx) \in G $,
		\item \label{itm:a5}For a power series $ p(\bt) = \sum_{m\geq 0} p_m\bt^{\otimes m} $, it holds that $ p\circ \alpha (\bt) = \alpha \circ p(\bt) $ for any $ \bt \in \Hopf $,
		\item \label{itm:invertible}
      If $ \bt \in \Hopf $ is invertible, then $\alpha(\bt^{-1}) = \alpha( \bt)^{-1} $,

		\item \label{itm:a4}Let $ f : \Hopf \to \R $ have the form 
		\begin{align}
		f(\bt) = \sum_m \sum_{i_1 \cdots i_m} f_m(\bt^{i_1 \cdots i_m}_m )
		\end{align}
   where we identify the degree-$m$ tensor $\bt_m \in (\R^d)^{\otimes m}$ of $\bt=(\bt_m)_{m \ge 0} \in \Hopf$ with its coordinates $\bt_m \simeq (\bt_m^{i_1,\ldots,i_m})_{i_1,\ldots,i_m=1,\ldots,d}$.
   If $ f_m(x) = f_m(-x) $ for every $ m\geq 0 $ and $x\in\R$, then
    \begin{align}
      f(\bt) = f(\alpha (\bt)).
    \end{align}
    	\end{enumerate}
\end{lemma}
\begin{proof}
	Items \ref{itm:a1} and \ref{itm:a2} follow from the definition.
  Item \ref{itm:group} is well known, see for instance~\cite[Section
  5]{chevyrev2016characteristic}.
  To see item \ref{itm:a5} note that
	\begin{align}
	p\circ \alpha (x) 
	= \sum_{n\geq 0} p_n(\alpha x)^{\otimes n} 
	= \sum_{n\geq 0} p_n\alpha (x^{\otimes n}) 
	= \alpha(\sum_{n\geq 0} p_n x^{\otimes n})
	= \alpha \circ p(x),
	\end{align}
	by Item \ref{itm:a2} and linearity.
  Item \ref{itm:invertible} follows since the inverse map $ \bt \mapsto \bt^{-1} $ has the power series expansion
	\begin{align}
    \bt^{-1} = \frac1{\bt_0} \Big\{\sum_{n\geq 0} (1-\frac1{\bt_0}\bt)^{\otimes n} \Big\},
	\end{align}
	see \cite[Lemma 7.16]{friz-victoir-book} which together with Item \ref{itm:a5} shows the claim. 
  For item \ref{itm:a4}, we note that
	\begin{align}
	f(\alpha x) 
	= \sum_n \sum_{i_1 \cdots i_n} f_n((-1)^nx^{i_n \cdots i_1})
	= \sum_n \sum_{i_1 \cdots i_n} f_n(x^{i_n \cdots i_1})
	= f(x).
	\end{align}
\end{proof}

A simple example of a function that satisfies the requirements Item \ref{itm:a4} in Lemma~\ref{lem:antipode} is the sum of squares
\begin{align}
	\Lee(\bt) = \sum_m \sum_{i_1 \cdots i_m} \vert \bt^{i_1 \cdots i_m}_m \vert^2
\end{align}
Which will be used to construct a loss function for tracks in Section \ref{sec:experiments}.

\paragraph{From Discrete Time to Continuous Time.}	
This section has so far focused on paths [resp.~tracks], that evolves [resp.~equivalence classes of evolutions] in continuous time $\bx: [0,T] \to \R^n$.
However, in practice one typically has only access to a discrete time observations $\bx(t_1),\ldots,\bx(t_L) \in \R^n$ along some grid $0 \le t_1 < \cdots < t_L \le T$, that is a time series.
But any TS can be identified as the piecewise linear path
\begin{align} \label{eq:disctime}
 t \mapsto \bx(t_i) + \frac{t_i -t}{t_{i+1}-t_i} (\bx(t_{i+1})- \bx(t_i))  \text{ for } t \in [t_i,t_{i+1})
\end{align}
and hence also as an element of $\Tracks$ after forgetting the parametrisation.
Working in continuous time when the original data is discrete might look cumbersome and unnecessary at first sight but it has several advantages.
Firstly, all TS are embedded into the same space $\Paths$ respectively $\Tracks$, even if the sample grid $t_1< \cdots t_L$ varies from TS to TS, which would not be the case if one identifies TS as vectors.
Secondly, this automatically ensures consistency in terms of high-frequency limits when the grid gets finer, that is $\sup |t_{i+1}^m-t_i^m| \to 0$ as $m \to \infty$ for a sequence $(t_i^m)_{i=1,\ldots,L_m}$.
Finally, many popular models are naturally formulated in continuous time rather than discrete time. 
\paragraph{From Tracks to Paths.}
Our guiding philosophy is that the fundamental object is the set $\Tracks$ rather than the set $\Paths$ since the former allows to ignore the time parametrization; note that the set of time parametrisations is infinite-dimensional since every continuous function $\tau: [0,T] \to [0,T']$ can be used to reparametrize a path $t \mapsto \bx(t)$ to $t \mapsto \bx(\tau(t))$, hence working with $\Tracks$ factors out an infinite-dimensional class of invariances.
Nevertheless, for certain applications the parametrization matters -- at least to a certain degree.  
However, this can be easily addressed by adding time as an additional coordinate: to emphasize the dimension $n$ of the state space $\R^n$ in which the paths evolve we write $\Paths_n$ (instead of just $\Paths$ that we used until now); similarly $\Tracks_n$ for the set of equivalence classes of $\Paths_n$.
Given $\bx \in \Paths_n$ we embed
\begin{align}
 \Paths_n \hookrightarrow \Paths_{n+1}, \quad \bx \mapsto \left(t \mapsto (t, \bx(t))\right).
\end{align}
That is a path evolving in $\R^n$ is turned into a path in $\R^{n+1}$ by simply adding an additional coordinate that is time itself.
This makes the parametrization part of the ``shape'' of the trajectory which in turn is exactly the information that distinguishes tracks, hence 
\begin{align}
  \Paths_n \hookrightarrow\Tracks_{n+1} . 
\end{align}
This injection shows that any scoring rule for tracks induces a scoring rule for paths.

	\section{Scoring Rules For Tracks and Paths}
  \label{sec:scoring tracks}
  Motivated by the toy example in Section~\ref{sec:scoring} with the moment feature map $\varphi$ for data in $\R^n$, we now follow the analogous reasoning on the non-Euclidean space of tracks by using the feature map $\Phi$ instead of $\varphi$, note that
  neither the domain nor the image of $\Phi$ is a linear space as its image is the group $G$ that is embedded into the linear space $\Hopf$. Recall that in Section~\ref{sec: structure of tracks} we have seen that it is exactly the group structure that captures the structure of space of tracks of concatenation and reversal.
  This motivates us to
  \begin{align}
\text{replace the additive inverse in } \varphi(X)-m \text{ by the group inverse to get } \Phi(\bX)m^{-1}.
  \end{align}
  Our first main result Theorem~\ref{thm:scoring} shows that this indeed leads to a proper scoring rule on the space of tracks and operations on $\Tracks$ turn into algebraic operations in decision space.
 Consquences of this result are Proposition~\ref{prop:entropy on tracks} and Corollary~\ref{cor:entrev} which show how the associated entropy on the space of tracks is invariant to time-reversal and behaves under conditioning on the past. 

 \paragraph{A Scoring Rule for Tracks.}
We need to introduce an additional space $ \HopfInv $ wedged between $ \Hopf $ and $ G $ defined as the space of all elements $ \bt \in \Hopf $ starting with a one, formally
\begin{align}
\HopfInv = \{ \bt \in \Hopf : \bt_0 = 1 \}.
\end{align}
Unlike $ \Hopf $, $ \HopfInv $ is not a vector space or a Hopf algebra, but it
is a group like $ G $ while also being convex as a subset of $ \Hopf $ in addition to being topologically closed -- unlike the set of invertible elements of $ \Hopf $.
We have the following sequence of inclusions
\begin{align}
G \subseteq \HopfInv \subseteq \text{ invertible elements of } \Hopf \subseteq \Hopf. 
\end{align}
	
	\begin{definition} \label{def:BA}
   Let $ \Lee : \Hopf \to \R $ be convex with a unique minimum at the unit $(1,0,0,\ldots)$ of $ G $.
   Define the left loss function as 
   \begin{align}
    \Le^\lt : \Tracks \times \Hopf^\star \to [0,\infty),\quad (\bx, m) \mapsto L(m^{-1} \Phi(\bx)). 
   \end{align}
  Applying step~\eqref{step: bayes act} from the scoring rule framework of Section~\ref{sec:scoring}, the left Bayes' act is defined as 
  \begin{align}
		\lB_\mu \definedas \underset{m \in \HopfInv_n}{\operatorname{argmin}}\,\E[\Le^\lt(m,\bX)]. %
  \end{align}
  Applying step~\eqref{step: scoring rule} yields the proper scoring rule
  \begin{align}
    s^\lt(\bx,\mu) \coloneqq \E[\Le^\lt(\lB_\mu, \bx)]. 
  \end{align}
  Applying step~\eqref{step: entropy} yields the (generalised) entropy,
  divergence, and mutual information 
  \begin{align}
    H^\lt(\mu) &\definedas \E_{\bX\sim \mu} \Le^\lt(\lB_\mu, \bX)\\
    d^\lt(\nu,\mu) &\definedas \E_{\bX\sim \nu}\Big[ \Le^\lt(\lB_\mu, \bX)- \Le^\lt(\lB_\nu, \bX) \Big]\\
    I^\lt(\mu,\nu) & \coloneqq H(\mu) - \E_{U \sim \nu}[ H(\mu|U)]
	\end{align}
on the output space $\cX= \Tracks$ and the action space $\cA= \Hopf^\star$.
  Analogously we define the right loss function $\Le^\rt(\bx,m) \coloneqq L(\Phi(\bx)m^{-1})$, right Bayes act $\rB_\mu$ and right scoring rule $s^\rt(\bx,\mu)$ as well as right entropy, divergence, and mutual information.
\end{definition}
The scoring rule framework of Definition~\ref{def:BA} turns operations on tracks into algebraic operations in the decision space. 
	\begin{theorem}\label{thm:scoring tracks}
    Let the output space be $\cX= \Tracks$, the action space $\cA= \Hopf^\star$ and
    \begin{align}
      \Le^\lt : \Tracks \times \Hopf^\star \to [0,\infty) \text{ resp. } \Le^\rt : \Tracks \times \Hopf^\star \to [0,\infty)
    \end{align}
    the loss functions from Definition \ref{def:BA}.
    The following properties hold
		\begin{enumerate} %
			\item\label{itm:c1} If $ \Lee $ is coercive, that is $ L(\bt)\to \infty $ whenever $ \lVert \bt \rVert \to \infty $, then for any Borel measure $ \mu $ such that $ \Lee $ is $ \mu $-integrable, both $ \lB_\mu $ and $ \rB_\mu $ exist. If $ \Lee $ is strictly convex, then they are unique.
			\item\label{itm:c2} If $ \mu = \delta_\bx $ then $ \rB_\mu = \lB_\mu = \Phi(\bx) $
			\item\label{itm:concat} The Bayes' acts satisfy 
			\begin{align}
			\lB_{\nu \star \mu \vert \nu} &= \nu(\Phi)\lB_{\mu\vert\nu} \\
			\rB_{\mu\star\nu \vert \nu} &= \rB_{\mu\vert\nu} \nu(\Phi) 			
			\end{align}
			where $ \nu(\Phi) $ denotes the pushforward measure of $ \nu $ under $ \Phi $ and by $\mu\star \nu| \nu$ we denote the law of $\bX \star \bY |\bY$ where $\Law{\bX}= \mu$ and $\Law{\bY}=\nu$.
			\item \label{itm:reversal}If $ \Lee $ satisfies $ \Lee(\bt) = \Lee(\alpha( \bt)) $, then 
			\begin{align}
			\rB_{\overleftarrow\mu} = \alpha(\lB_\mu), \quad 
			\lB_{\overleftarrow\mu} = \alpha(\rB_\mu)
			\end{align} 
			where $ \overleftarrow\mu $ denotes the measure $ \mu $ given by running samples from $\mu$ backwards in time\footnote{Formally $\bX \sim \mu$ then $\overleftarrow \mu$ is defined as the law of $\overleftarrow \bX$.}
		\end{enumerate}
	\end{theorem}
	\begin{proof} We give the proofs for the right Bayes' act as the proofs for the left Bayes' act is similar.
		
	(\ref{itm:c1})
	We equip $ \Hopf_n$ with its $ \ell^2 $ norm,
	\begin{align}
	\lVert \bt \rVert = \sqrt{ \sum_n \vert \bt_n \vert^2 }
	\end{align}
	which makes it into a separable Hilbert Space.
	
	Fix some measure $ \mu $ on $ \cX $ and define the map $ \psi : G_n \to \R $ by
	\begin{align}
	\psi(x) = \EE_\mu \Lee(\Phi(\bX)x).
	\end{align}
	We want to show that $ \psi $ is convex, coercive and lower semicontinuous on $ (\Hopf_n, \lVert \cdot \rVert) $, as this guarantees the existence of a minimiser. This is because the unit ball of $ (\Hopf_n, \lVert \cdot \rVert) $ is weakly compact, hence we could choose some weakly compact and convex set $ C $ such that $ \psi(x) > M $ outside of $ C $, and since $ \psi $ is lower semicontinuous it is also weakly lower semicontinuous, and therefore since it is convex it achieves a minimum on $ C $ which must be a global minimum. Note that its minimiser must be $ (\rB_\mu)^{-1} $. It follows that if $ \psi $ is strictly convex, then the minimiser is unique.
	
	Note that if $ \Lee $ is (strictly) convex, then
	\begin{align}
	\psi(\frac12x+\frac12y) 
	= \EE_\mu \Lee(\frac12\Phi(\bX)x+\frac12\Phi(\bX)y)
	\leq\!\!(<)\,\,\, \frac12\EE_\mu \Lee(\Phi(\bX)x)+\frac12\EE_\mu \Lee(\Phi(\bX)y) 
	= \frac12\psi(x) + \frac12\psi(y)
	\end{align} 
	hence $ \psi $ is also (strictly) convex.
	
	Note that $ ( \Hopf_n, \lVert \cdot \rVert ) $ is a Banach algebra, that is $ \lVert \bt \cdot \bs \rVert \leq \lVert \bt \rVert \cdot \lVert \bs \rVert $. By taking multiplicative inverses, this implies that
	\begin{align}
	\lVert \bt \cdot \bs \rVert \geq \frac1{\lVert \bs^{-1} \rVert}\lVert \bt \rVert.
	\end{align}
	for any invertible element $ \bs $. As $ \mu $ is Borel, and $ \Hopf_n $ is a Polish space, $ \mu $ is a Radon measure by \cite[Theorem 7.1.7]{Bogachev07} and we may choose a compact set $ C $ such that $ \mathcal P(\bX \notin C) \leq \varepsilon $, and define $ M $ to be $ \sup_{X \in C} \lVert \Phi(X)^{-1} \rVert $. 
	Then on $ C $, $ \lVert \Phi(X)x \rVert \geq \frac1M \lVert x \rVert $, and since 
	\begin{align}
	\psi(x) = \EE_\mu \Lee(\Phi(\bX)x) \geq \EE_{\mu\vert_C} \Lee(\Phi(\bX)x)
	\end{align}
	and $ \Lee $ is coercive, so is $ \psi $.
	
	To see that $ \psi $ is lower semicontinuous, note that for a sequence $ x_k \to x $ 
	\begin{align}
	\lim\inf_k \psi(x_k) 
	= \lim\inf_k \EE_\mu \Lee(\Phi(\bX)x_k) 
	\geq \EE_\mu \Lee(\Phi(\bX)x)
	= \psi(x)
	\end{align}
	by Fatous Lemma, the assertion follows. 
	
	(\ref{itm:c2})
	Since $ \Lee $ is minimised at the unit it is clear that for $ \mu = \delta_\bx $, $ \lB_\mu = \rB_\mu = \Phi(\bx) $ is optimal since $ (\lB_\mu)^{-1}\Phi(\bx) = \Phi(\bx)(\rB_\mu)^{-1} = 1 $.
	
	(\ref{itm:concat}) For $ \bY \sim \nu $ we have
	\begin{align}
	\rB_{\mu\star\nu \vert \nu} \definedas 
	&\underset{m \in \TAf{\R^n}}{\operatorname{argmin}}\E_{\bX\sim\mu}\big[\Lee(\Phi(\bX\star \bY)m^{-1}) \vert \bY \big] = \\
	&\underset{m \in \TAf{\R^n}}{\operatorname{argmin}}\E_{\bX\sim\mu}\big[\Lee(\Phi(\bX)\Phi(\bY)m^{-1}) \vert \bY \big] = \\
	&\underset{m\Phi(\bY)^{-1} \in \TAf{\R^n}}{\operatorname{argmin}}\E_{\bX\sim\mu}\big[\Lee(\Phi(\bX)m^{-1}) \vert \bY \big] \definedas
	\rB_{\mu \vert \bY} \Phi(\bY).
	\end{align}
	
	(\ref{itm:reversal}) 
	If $ \Lee $ satisfies $ \Lee(\bt) = \Lee(\alpha( \bt)) $, then
	\begin{align}
	\rB_{\overleftarrow\mu} \definedas 
	&\underset{m \in \TAf{\R^n}}{\operatorname{argmin}}\E_{\bX\sim \overleftarrow\mu}\Lee(\Phi(\bX)m^{-1}) =  \\
	&\underset{m \in \TAf{\R^n}}{\operatorname{argmin}}\E_{\bX\sim \mu}\Lee(\alpha(\Phi(\bX))m^{-1}) = \\
	&\underset{m \in \TAf{\R^n}}{\operatorname{argmin}}\E_{\bX\sim \mu}\Lee(\alpha(m^{-1})\Phi(\bX)) \definedas 
	\alpha(\lB_{\mu}).
	\end{align}
	\end{proof}
  \paragraph{Entropy, Divergence, and Mutual Information on the Space of Tracks.}
  We now focus on the (generalized) entropy, divergence and mutual information for probability measures on tracks that results from Definition~\ref{def:BA}.  

	\begin{proposition}\label{prop:entropy on tracks}
    For any two probability measures $\mu$ and $\nu$ on $ \Tracks$ it holds that %
		\begin{enumerate}
			\item \label{itm:b1}$ H^{\lt}(\mu \star \nu \vert \mu) = H^\lt(\nu \vert \mu) $ and $ H^{\rt}(\mu \star \nu \vert \nu) = H^\rt(\mu \vert \nu) $
			\item \label{itm:b2}If $ \Lee $ satisfies $ \Lee(\bt) = \Lee(\alpha \bt) $, then
			\begin{align}
			H^{\rt}(\mu) &= H^{\lt}(\overleftarrow\mu), \\
			d^{\rt}(\nu,\mu) &= d^{\lt}(\overleftarrow\nu,\overleftarrow\mu), \\
			I^{\rt}(\mu,\nu) &= I^{\lt}(\overleftarrow\mu,\nu)
			\end{align}
		\end{enumerate}
	\end{proposition}
	\begin{proof}
		(\ref{itm:b1}) \begin{align}
		H^{\lt}(\mu \star \nu \vert \mu) &= \E_{X \sim \mu \star \nu \vert \mu} \Lee((\lB_{\mu \star \nu \vert \mu})^{-1}\Phi(X)) \\
		&= \E_{X \sim \nu \vert \mu} \Lee((\lB_{\nu \vert \mu})^{-1}\nu(\Phi)^{-1}\nu(\Phi)\Phi(X))
		= H^{\lt}(\nu \vert \mu).
		\end{align}
		(\ref{itm:b2}) \begin{align}
		\E_{X \sim \nu} \Lee(\Phi(X)(\rB_{\mu})^{-1}) 
		= \E_{X \sim \nu} \Lee(\alpha (\rB_{\mu})^{-1}\alpha\Phi(X))
		= \E_{X \sim \overleftarrow\nu} \Lee((\lB_{\overleftarrow\mu})^{-1}\alpha\Phi(X)).
		\end{align}
		The other equalities follow.
	\end{proof}
	\begin{corollary} \label{cor:entrev}
		If $ \Lee $ satisfies $ \Lee(\bt) = \Lee(\alpha (\bt)) $ and a measure $ \mu $ is \emph{reversible}, that is, $ \mu $ and $ \overleftarrow\mu $ are equal up to their starting distribution, then
    \begin{align}
      H^{\rt}(\mu) = H^{\lt}(\mu).
    \end{align}
	\end{corollary}
		In the experiments we will simulate sample paths from Brownian motion and since this is a reversible process it will not matter if we use the left- or right entropy.
  \begin{remark}\label{rem:flatness}
    An alternative to using the group structure in Definition~\ref{def:BA} of the Bayes act is to use that the group $G$ is embedded into the linear space $\Hopf$ and use this linear structure.
    That is, we define a Bayes act as $a_\mu \coloneqq \operatorname{argmin}_{m \in \TAf{\R^n}}\E_{\bX \sim \mu}[ \Lee(\Phi(\bX) - m)^2]$.
    It is easy to show that this gives a proper scoring rule and that $a_\mu = \E[\Phi(\bX)]$.
    However, this scoring rule relies on the embedding of the group into its ambient vector space and does not account for or respect the group structure.
    Moreover, the resulting divergence and entropy reduce to just the Euclidean distance and usual variance.
    The same remark extends to (signature) kernel based scoring, where linear methods are used in an RKHS; see the discussion about non-kernel based scoring in the introduction.     %
  \end{remark}
  \begin{remark}
    We identify a stochastic process as a path- or sequence-valued random variable, possibly even ignoring its parametrization.
    However, for some applications the filtration of a stochastic process matters and one could ask to extend the scoring rule framework to this.
    A kernel that captures the filtration was introduced in~\cite{bonnier2021adapted} and a kernel algorithm and new applications given in \cite{salvi2021higher}.
    To get a non-kernel scoring one could try to replace $\Phi$ in Definition \eqref{def:BA} by the higher-rank signature from~\cite{bonnier2021adapted}.
  \end{remark}
	\section{Gradient Descent on the Space of Tracks}
  \label{sec:gradients}

  Given a smooth function $f: \R^n \to \R$, the simplest update rule for gradient descent is  
	\begin{align}\label{eq:classic gradient}
	x_{i+1} = x_i -\eta \nabla f(x_i)
	\end{align}
  and under additional regularity of $f$, the resulting sequence $(x_i)_i \subset \R^n$ converges to a minimizer of $f$. 
  Our interest lies minimizing functions $f: \Tracks \to \R$.
  In accordance with our guiding theme we do not identify these domains as linear spaces where classical gradient descent can be applied. 
  However, we have seen that $\Phi$ provides an isomorphism between the space of tracks and the free group $G$ (up to forgetting the starting point of the track)
  \begin{align}
   \Tracks \simeq G 
  \end{align}
  Hence, the minimization problem of a function $f$ of tracks can be re-formulated as a minimization problem of a function $F=F(g) \coloneqq f \circ \Phi^{-1}(g)$ on the free group $G$.
  That is, the general problem we try to solve is to find 
  \begin{align}
   \operatorname{argmin}_{g \in G}F(g) \text{ for }F: G \to \R 
  \end{align}
  for any $F$ in class of sufficiently ``smooth'' real-valued functions on $G$. 

	There have been many attempts to generalize gradient descent to non-linear domains.
  Arguably the the case of Riemannian manifolds~\cite{Bonnabel2013} is the most well-developed among these.
  However, the group $G$ does not come with a Riemannian structure (to wit, only a Sub-Riemannian structure \cite{montgomery-02}). 
  We follow here a somewhat different approach inspired by work of Pierre Pansu \cite{pansu1989metriques} that directly uses the group structure to define gradients.
  We show that this gradient in turn allows us to give a straightforward generalization of the gradient update rule~\eqref{eq:classic gradient} from $\R^n$ to $G$, so that the resulting sequence $ (g_i)_i \subset G$ converges to the minimizer. 
\paragraph{Pansu Derivatives.}
The derivative $ D f (x)(\cdot)$ of a function 
\begin{align}
 f : \R^n \to \R 
\end{align}
at a point $x$ is a linear functional of $\R^n$ that can be defined as the limit 
\begin{align}
 D f(x)(h) \coloneqq \lim_{h \to 0} \frac{f(x+\lambda h) - f(x)}{\lambda}.
\end{align}
Identifying $\R^n$ as a the additive group $(\R^n, + )$ one can regard the difference quotient that appears in the limit as applying to the group operation to $x$ and $\lambda h$.
Hence, if we have also have a generalization of the multiplication with the scalar $\lambda$, then the above difference quotient makes sense for other groups than the additive group $(\R^n, +)$.
To formalize scalar multiplication, it turns out that the right notion is that of a Carnot group: a Carnot group is a Lie group $C$ that carries a left-invariant geodesic distance $\operatorname{dist}:C \times C \to [0,\infty)$ and for each $\lambda >0$ a bijection
\begin{align}
  \delta_\lambda : C \to C \text{ such that } \operatorname{dist}\left(\delta_{\lambda}(g), \delta_\lambda(h)\right) = \lambda \operatorname{dist}\left(g, h\right).
\end{align}
However, our focus is on the free group $G$ and here the scaling $\delta_\lambda$ by a scalar $\lambda>0$ has an explicit form
\begin{proposition}
  The group $G \subset \Hopf$ equipped with the geodesic distance and 
  \begin{align}
 \delta: \lambda \times G \to G, \quad   \delta_\lambda g = \delta_\lambda(1, g_1, g_2, g_3, \ldots) = (1, \lambda g_1, \lambda^2g_2, \lambda^3g_3, \ldots).
  \end{align}
  forms a limit of Carnot groups.
\end{proposition}
We now have all we need to define the (Pansu) derivative. We denote by $ G^\star $ the topological dual of $ G $.
\begin{definition}
  Let $ f : G \to \R $.
  We define $ Df : G \to G^\star $ as 
  \begin{align}
		Df(g)h = \lim_{\lambda\downarrow 0} \frac{f(g\delta_\lambda h) - f(g)}{\lambda} 
  \end{align}
  whenever this limit exists and call $Df(g)h$ the Pansu derivative of $f$ at $g$ in direction $h$.
  If $f$ has a Pansu derivative for all $g \in G$ then we say that $f$ is Pansu differentiable. Analogous we define the spaces $C^k(G,\R)$ of $k$-times Pansu differentiable functions.
\end{definition}
The Pansu derivative behaves very similar to the classic linear gradient.
For example, for the proof of convergence of gradient descent on $G$ we make use of the following ``Taylor expansion''.

	\begin{lemma} \label{lem:taylor}
		If $ f \in C^3(G,\R) $, then 
		\begin{align}
		f(gh) = f(g) + Df(g)h_1 + \frac12 D^2f(g)h_2 + O(\lVert h_3 \rVert)
		\end{align}
	\end{lemma}
	\begin{proof}
		Consider the function $ A : \R \to \R $
		\begin{align}
		A(\lambda) = f(g\delta_\lambda h)
		\end{align}
		then $ A $ is $ C^3 $ and by a (classical linear) Taylor expansion we may write
		\begin{align}
		A(1) = A(0) + \dot A(0) + \frac12 \ddot A(0) + O(\lVert A^{(3)}(0) \rVert )
		\end{align}
		which translates into the asserted equation since $ h^{\otimes n} $ is contained in $ h_n $.
	\end{proof}

  \begin{remark}
   A popular approach to differentiating functions of paths is to use a Fr\'echet derivative as in Malliavin calculus, i.e.~one identifies the space of paths as a linear space, see \cite{CASS2011542}.
   However, the above Pansu derivative is of a very different nature and -- by construction -- respects the non-Euclidean structure of paths resp.~tracks.
 \end{remark}
\paragraph{Gradient Descent on $G$.}
The idea of gradient descent to take a step in a direction that minimises $ f $ in a neighbourhood $ B_h $.
In our (Lie) group $G$, a natural choice is to choose some vector $ v \in \R^n $ and use an exponential neighbourhood $ g e^{\eta v} $ of $ g \in G $ where $\exp$ denotes the exponential from Lie algebra to Lie group.
Hence, the question becomes how choose $ v $ to minimise $ f(ge^{\eta v}) $.
To do so, note that
	\begin{align}
	f(ge^{\eta v}) = f(g\delta_\eta e^v)
	\end{align}
	whenever $ v \in \R^n $. Now using Lemma \ref{lem:taylor} shows 
	\begin{align}
	f(g\delta_\eta e^v) = f(g) + Df(g) e^{\eta v}_1 + \eta^2 = f(g) + \eta Df(g)\cdot v + \eta^2.
	\end{align}
	This suggests that the direction of steepest descent in the exponential neighbourhood is indeed given by the Pansu derivative (henceforth, and with slight abuse of notation, we identify the resulting element of $ G^\star $ as an element of the Hilbert space $ \Hopf $ since $G$ and $G^*$ both embed into $\Hopf$), %
  \begin{align}
    v = - Df(g).
  \end{align}
  This leads to the following geometric descent rule
	\begin{align} \label{eq:update}
	g_{i+1} = g_i e^{ -\eta D f(g_i)}.
	\end{align}
  As for classic gradient descent, we need a notion of convexity to guarantee convergence to a minimum. 
	\begin{definition}
		We say that a function $ f : G \to \R $ is geometrically convex if
		\begin{align}
      f(g\delta_\lambda (g^{-1}h)) \leq (1-\lambda)f(g) + \lambda f(h)
		\end{align}
		for any $ 0 \leq \lambda \leq 1 $.
	\end{definition}
	As in the linear case, one can show that if $ f : G \to \R $ is geometrically convex and $ C^2 $, then $ D^2f(g) $ is positive definite everywhere.
  However, $ D^2 f $ is not symmetric in general unlike in the linear case.
Putting everything together allows us to mimic the convergence proof of gradient descent in linear spaces which ultimately justifies the above informal derivation of the update rule.
  \begin{theorem}\label{thm:gradient descent}\,
    \begin{enumerate}
    \item \label{itm: transitive}
    The geometric update rule is transitive, that is, one may go from any point in the group to any other point using updates of the form \eqref{eq:update}.
\item\label{itm: convergence} 
  If $ f $ is geometrically convex and bounded from below with bounded second Pansu derivatives, then for $ \eta $ sufficiently small the sequence in Equation \eqref{eq:update} converges to a minimum
  \end{enumerate}
\end{theorem}

\begin{proof}
  Item \ref{itm: transitive} follows from Chow's theorem~\cite{Chow1940} which states that $ G $ is generated by simple exponentials.
  For Item~\ref{itm: convergence} let $ L $ be a bound on the derivatives of $ f $.
  We may write
  \begin{align}
		f(gh) \leq f(g) + Df(g)h_1 + \frac12 L\lVert h_2\rVert.		
  \end{align}
  Hence, if $ g_{n+1} = g_n e^{ -\eta D f(g_n)} $, then
  \begin{align}
		f(g_{n+1}) &\leq f(g_n) - \eta \lVert Df(g_n) \rVert^2 + \frac12 \eta^2 L\lVert Df(g_n) \rVert^2 = f(g_n) + (\frac{L}2\eta^2 - \eta)\lVert Df(g_n) \rVert^2
  \end{align}
  which is smaller than $ f(g_n) $ for $ \eta $ small enough whenever $ Df(g_n)\not= 0 $. Hence this is a strictly decreasing sequence and converges to a minimum.
\end{proof}

	Theorem \ref{thm:gradient descent} gives the analogous convergence guarantees as regular gradient descent and is simple to implement.

	\section{Experiments}\label{sec:experiments}
	
	In all our experiments we take as loss function $ \Lee $ the squared norm,
	\begin{align}
	\Lee(\bt) = \lVert \bt \rVert^2 = \sum_{m=1}^M \lVert \bt_m \rVert^2.
	\end{align}
  which is symmetric by Lemma \ref{lem:antipode}. 
  Note that the function $ \Lee $ is smooth as a sum of squares, and since $ \Phi $ is also smooth, since it is polynomial in the increments of its input when computed up to some fixed degree $ M $, we can easily compute $ d, H, $ and $ I $ by automatic differentiation. For the computation of $\Phi$ we use the signatory~\cite{kidger2021signatory} package which allows for fast and easy computations. 
	\begin{remark}
		In the experiments we will simulate sample paths from Brownian motion. As this is a reversible process it does not matter if entropy is computed from the left or the right by Corollary \ref{cor:entrev}, and we will normally compute it from the right. 
	\end{remark}

	\subsection{Comparing a warped time-series to itself}
	
	One natural question to ask is how well the regularised versions of DTW that allow for differentiation are able to incorporate the parametrisation invariance.
  One drawback of these regularisation of DTW is that it typically leads to a trade-off between the smoothness and parametrization invariance. 
	
  Given two TS $\bx$ and $\by$ we then use sDTW $d_{\text{sDTW}}(\bx,\by)$ which is the regularised version of DTW introduced in \cite{Blondel2021DifferentiableDB} which is differentiable and to a certain degree parametrisation invariant quantity.
  The scoring rule framework that we presented in the previous sections provides divergence not only between TS but probability measures on TS,
	\begin{align}
	d(\nu,\mu) &\equiv\E_{X\sim \nu}[ S(X,\mu)]-H(\nu),
	\end{align}
  but as a special case, we can restrict this to point measures to get a ``divergence'' between two TS akin to sDTW which reduces to the formula
	\begin{align}\label{eq:sig divergence}
	(\bx,\by) \mapsto d(\delta_{\bx},\delta_{\by}) = \Lee(\Phi(\bx)\Phi(\by)^{-1}).
	\end{align}
  To study the different behaviour between these two ``divergences'' -- sDTW and \eqref{eq:sig divergence} -- we generated samples paths from a Brownian motion to get $\bx$ and warped it a time change $ \varphi : [0,T] \to [0,T], \varphi(t) = T(t/T)^p $ to get $\by$ where $ p \in [1, \infty) $ is a parameter that determines the severity of the warping.
  Finally, we sampled on a discrete time grid with resolution $ 10^{-2} $.
  The results are shown in Figure \ref{fig:warp_comp}.
  The geometric score stays close to $ 0 $ regardless of the value of $ p $ but the sDTW divergence will increase with $ p $ and the rate of increase depends on the $ \gamma $ parameter.
  As $ \gamma $ tends to $ 0 $ the sDTW score will tend to the geometric score, but it will lose its smoothness while doing so.
  It is worth noting that $ \gamma = 1 $ is considered a default value.
	To achieve the same level of invariance enjoyed by the the divergence \eqref{eq:sig divergence} in a DTW paradigm on can use the classical non-smooth DTW algorithms which cannot be updated using gradient descent.
	In contrast, the signature divergence \eqref{eq:sig divergence} is always differentiable and parametrization invariant.
  Further, the signature divergence is more general in the sense that it is not just a divergence between TS but between probability measures on TS which allows in principle for many other applications such as variational inference. 
	\begin{figure}[H]
    \centering
		\includegraphics[width=4.1in]{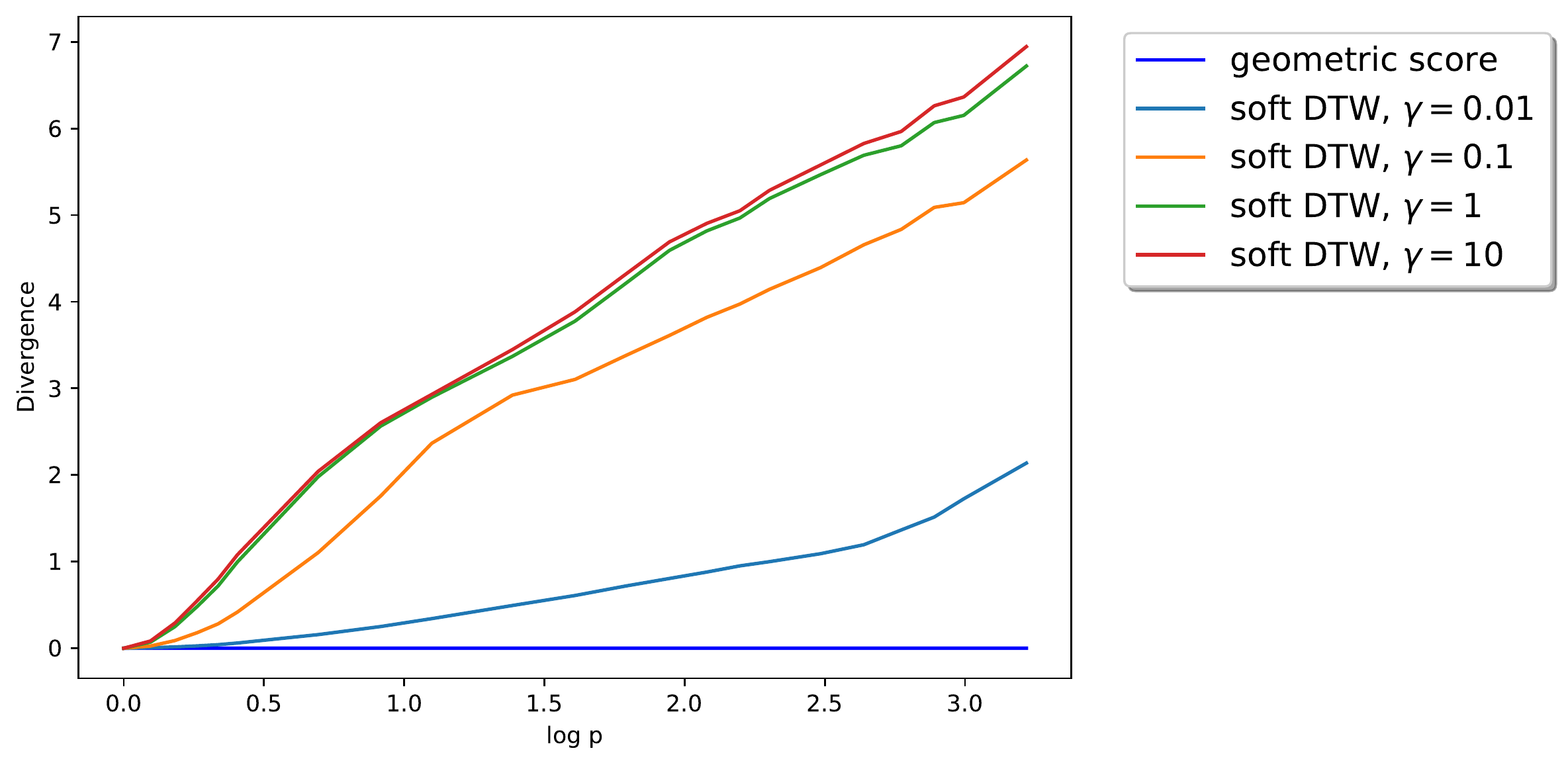}
		\caption{Geometric scoring compared to sDTW for different values of $ \gamma $ plotted against the logarithm of $ p $ for $ p $ between $ 1 $ and $ 25 $. }
		\label{fig:warp_comp}
	\end{figure} 
  The implementation of sDTW is taken from the excellent Python package from \cite{Blondel2021DifferentiableDB}.

	\subsection{Mutual information}
	Recall that the mutual information between two probability measures $ \mu $ and $ \nu $ is defined as
	\begin{align}\label{eq:MI}
	I(\mu,\nu) & \coloneqq H(\mu) - \E_{U \sim \nu}[ H(\mu|U)].
	\end{align}
  and provides a dependency measure between $\mu$ and $\nu$.
  Closest related to independence on paths are the \emph{signature cumulants}
  \cite{Bonnier2020} which have been proven to be useful in applications
  \cite{schell2021nonlinear, friz2021unified}.
  However, signature cumulants only compare probability measures on paths [tracks] with other probability measures on paths [tracks]. 
  In contrast, the mutual information \eqref{eq:MI} allows to measures dependency between a probability measure $\mu$ on paths [tracks] and a probability measure $\nu$ on an arbitrary measurable space; in particular, this allows to measure dependence relations between TS and scalar-valued random variables.
  
  \begin{figure}[H]
  	\centering
  	\includegraphics[align=t, scale=0.4]{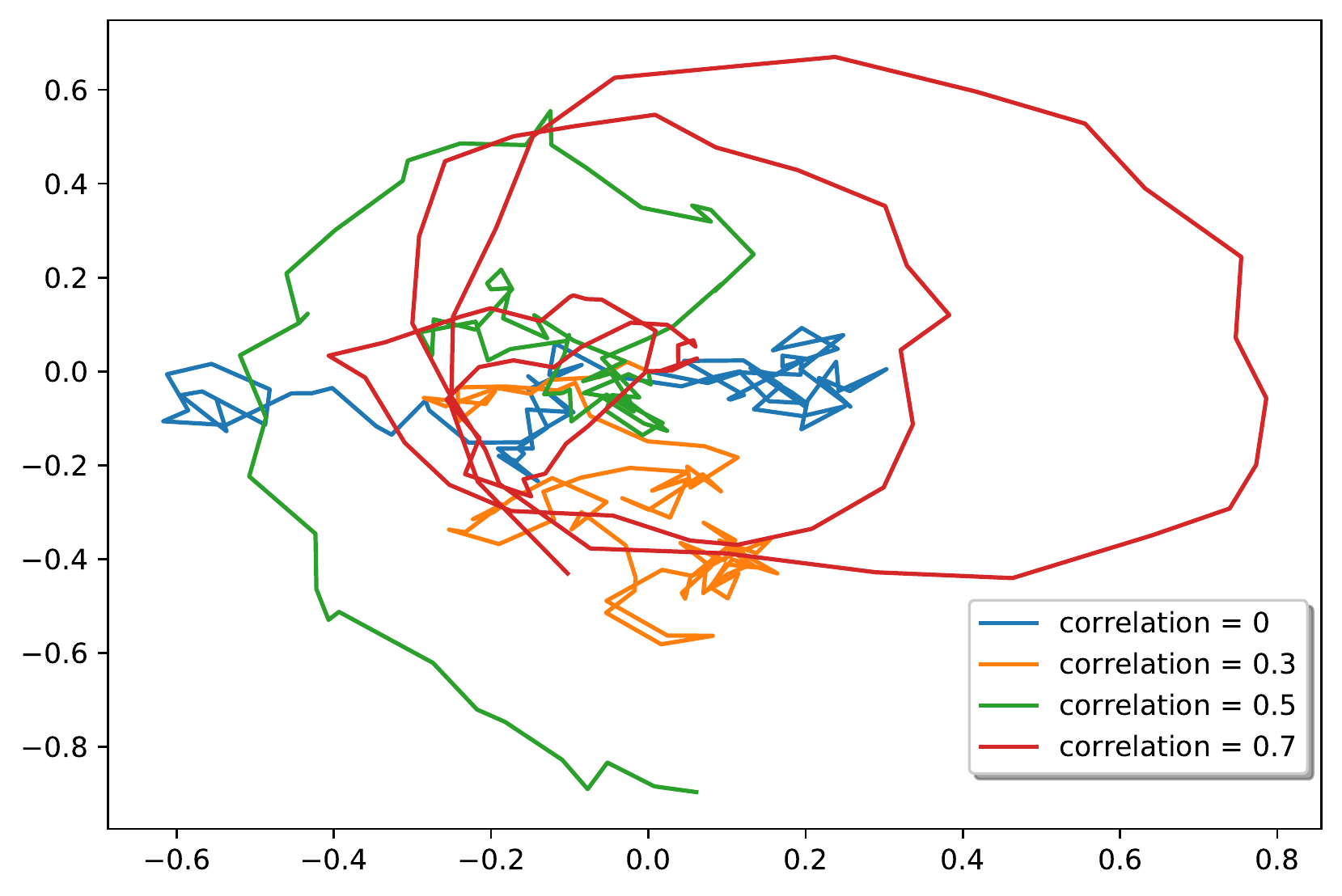}
  	\includegraphics[align=t,scale=0.4]{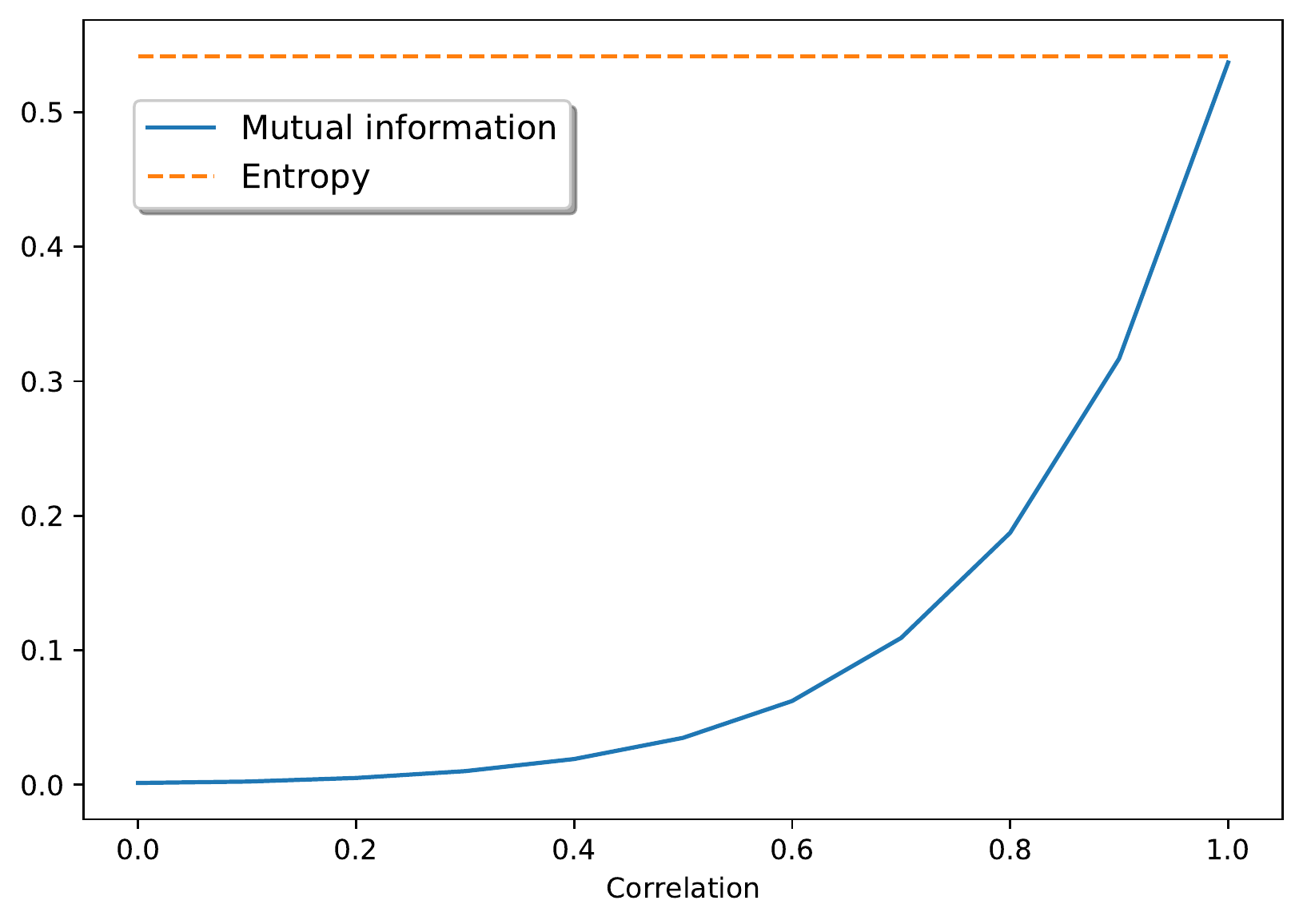}
  	\caption{The left hand side shows sample trajectories from $\bY$ for different correlation values $\rho$. The right plot shows the mutual information between $\bY$ and the rotational speed $\omega$ as a function of $\rho$.}
  	\label{fig:shift_MI}
  \end{figure} 
  
  To demonstrate this, we consider two examples
  \begin{description}
  \item[Mutual information between a stochastic process and a scalar.] 
   We consider the stochastic process $\by$ defined as 
    \begin{align}
      \by_t = \rho t {\cos \omega t \choose \sin \omega t} + \sqrt{1-\rho^2} \bx_t
    \end{align}
    where $\omega$ is sampled from the uniform distribution on $[0, 8\pi) $ and $\bx$ is a standard $ 2 $-dimensional Brownian motion independent from $ \omega $. $ \bx $ and $ \by $ are generated on the interval $ [0,1] $ on a grid with resolution $ 10^{-2} $.
    The right plot in Figure \ref{fig:shift_MI} shows some sample trajectories of the process $\by$.
    The left plot in Figure \ref{fig:shift_MI} shows the mutual information between (the law of) $ \omega $ and $ \by $ for various values of $\rho$.
    As expected, the mutual information increases monotone with $\rho$ and is bounded by the entropy.
    
	\item[Mutual Information between two unparametrized stochastic processes.] 
    We consider two independent a Brownian motions $ \bx $ and $ \by $ as before, and a random parametrisation $ \phi(t) = T(t/T)^p $ where $ p \in [1,10] $ is uniformly distributed independent of $ \bx $. The stochastic process $ \bz $ is defined as
	\begin{align}
    \bz_t = \rho \bx_{\phi(t)} + \sqrt{1-\rho^2} \by_t.
	\end{align}
	We then compute the mutual information between $ \bz $ and $ \bx $ as a function of $ \rho $. The results are shown in Figure \ref{fig:warp_MI}. 

  \end{description}

  \begin{figure}[H]
    \centering
    \includegraphics[align=t, scale=0.4]{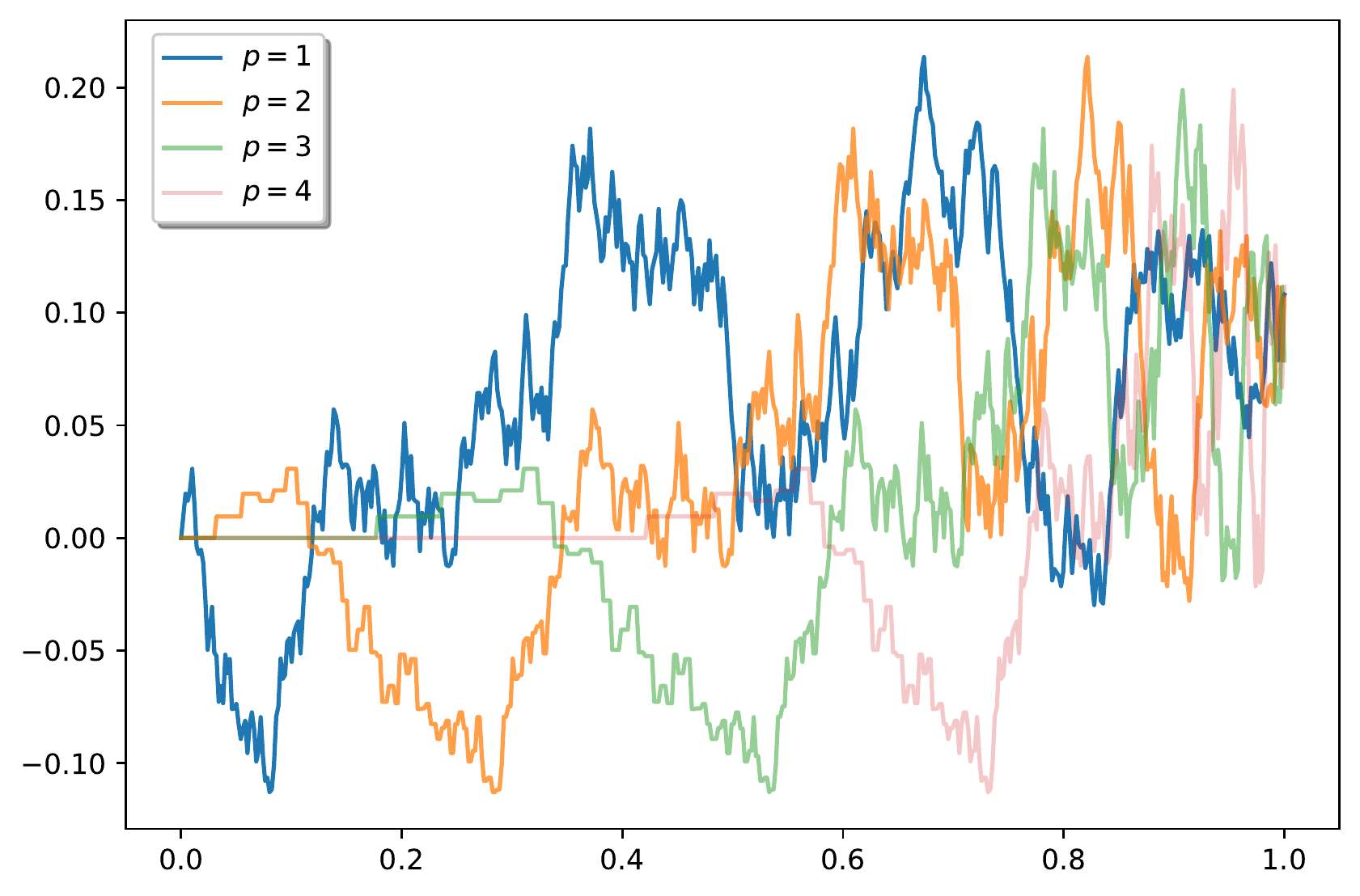}
    \includegraphics[align=t,scale=0.4]{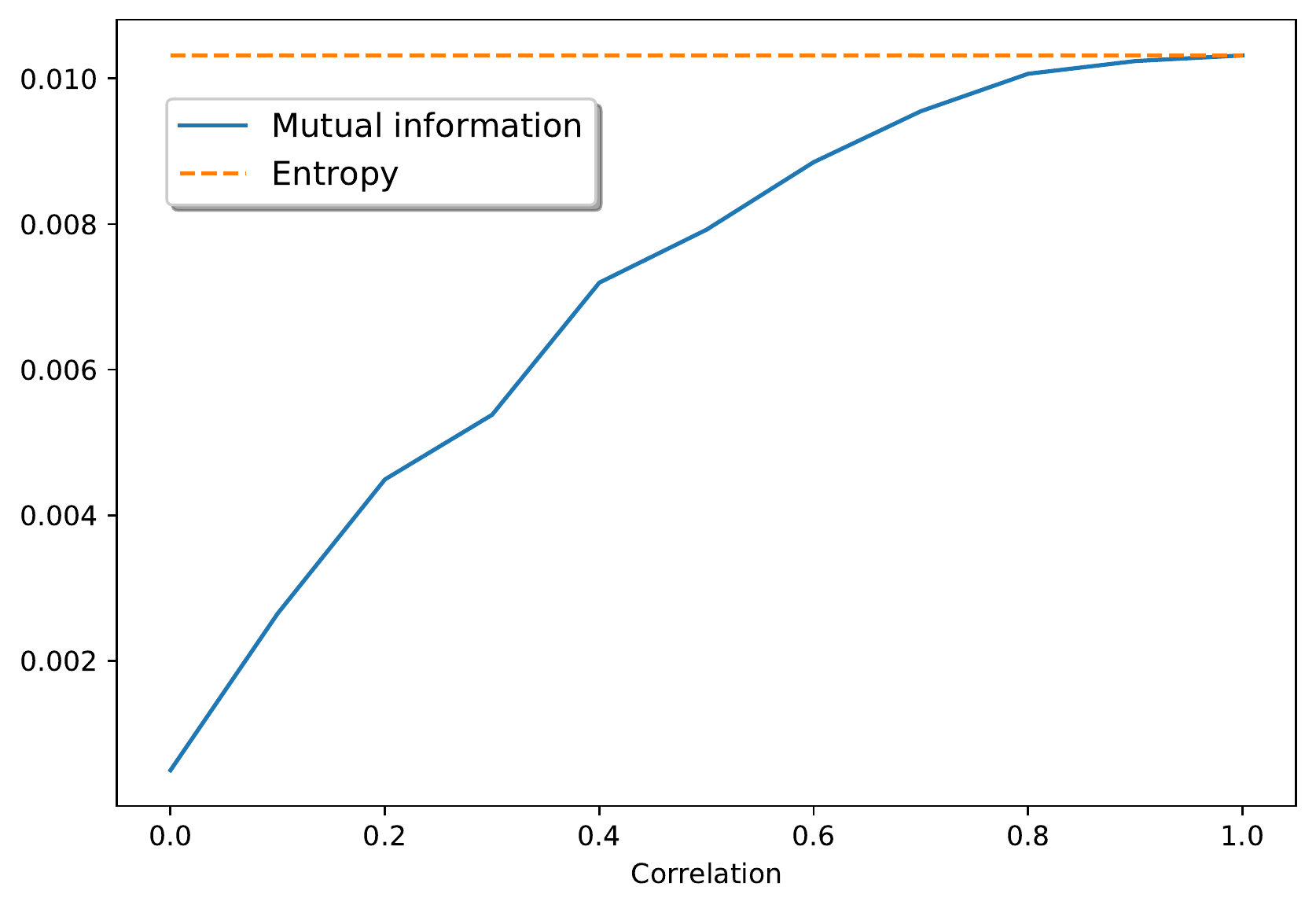}
    \caption{The left plot shows one-dimensional projections of sample paths of $ \bx $ for different warping parameters $ p $. The right plot shows the mutual information between $ \bx $ and $ \bz $ as a function of $ \rho $.}
		\label{fig:warp_MI}
		\label{fig:warp_plot}
	\end{figure} 
	
	\subsection*{Acknowledgements}
	PB is supported by the Engineering and Physical Sciences Research Council [EP/R513295/1]. 
	HO is supported by the EPSRC grant ``DATASIG'', the Turing Institute, and the Oxford-Man Institute of Quantitative Finance.
	
	\newpage	
	\appendix
	\section{Tree-like equivalence of paths}\label{app:TLE}

	Informally, a path $ x : [0,T] \to \R^n $ is \emph{tree-like} if the set it traces out in $ \R^n $ looks like a tree.
  The formal definition is 
	\begin{definition}[\cite{hambly-lyons-02}]
			A continuous path $ x : [0,T] \to E $ is said to be \emph{tree-like} if there exists an $ \R $-tree $ \tau $, a continuous map $ \varphi : [0,T] \to \tau $ and a map $ \psi : \tau \to E $ such that $ \varphi(0) = \varphi(T) $ and $ x = \psi \circ \varphi $.
      Two paths $ \bx, \by $ are \emph{tree-like equivalent} if $ \bx \star \by^{-1} $ is tree-like and we denote this relation with $\bx \sim \by$.
	\end{definition}
	That is, if $ \bx $ and $ \by $ follow the same trajectory in $ \R^n $ up to tree-like excursions then $\bx \sim \by$.
  In particular, if $\by$ is just $bx$ under time-reparametrization, $\by(t)=\bx(\tau(t))$ for an increasing $\tau$, then $\bx \sim \by$.
	See also~\cite[Appendix B]{Bonnier2020} for more details.
  It is easy to check that $\sim$ is an equivalence relation, hence we can define  
	\begin{align}
	\tracks{a}{b} \coloneqq \paths{a}{b} / {\sim} \text{, and } \Tracks\coloneqq \Paths / {\sim},
	\end{align}

	\bibliographystyle{unsrt}
	\bibliography{roughpaths}

  \end{document}